%% file: main.tex
\documentclass[10pt]{article}
\usepackage{amsthm,amsfonts,amsmath,amssymb}
\usepackage{geometry}
\usepackage{bbm}		 			
\usepackage{float}					
\usepackage{color}
\definecolor{toc}{RGB}{13,55,174}	
\usepackage[hidelinks]{hyperref}	
\hypersetup{
    colorlinks=true,
    citecolor=red,
    filecolor=black,
    linkcolor=toc,
    urlcolor=toc
} 
\usepackage{bm}						
\usepackage{enumitem}				

\usepackage{adjustbox}
\allowdisplaybreaks					
\usepackage{subcaption} 
\usepackage{mathtools} 

\usepackage{algorithm}
\usepackage{algcompatible}

\usepackage[symbol*]{footmisc}
\usepackage{amsfonts}
\usepackage{amsmath}
\usepackage{bbm}
\usepackage{bm}
\usepackage{thm-restate}
\newtheorem{theorem}{Theorem}[section]
\newtheorem{lemma}[theorem]{Lemma}
\newtheorem{corollary}{Corollary}[theorem]

\newcommand{\e}{\varepsilon} 
\newcommand{\lp}{\left} 
\newcommand{\rp}{\right}
\renewcommand{\Pr}[2]{\textbf{Pr}_{#1}\lp[#2\rp]} 
\newcommand{\E}[1]{\mathbb{E}\lp[#1\rp]}

\newcommand{\oracle}{\mathcal{O}}
 
\newtheorem{definition}{Definition}

\newcommand{\noisygraph}{\textsc{graph connectivity with noisy queries}}

\newcommand{\pebe}{This work is partially supported by NTUA Basic Research Grant (PEBE 2020) ``Algorithm Design through Learning Theory: Learning-Augmented and Data-Driven Online Algorithms (LEADAlgo)''.}

\title{Graph Connectivity with Noisy Queries\footnotetext[0]{\pebe}}
\author{
Dimitris Fotakis$^1$  \\ {\tt fotakis@cs.ntua.gr} \and
Evangelia Gergatsouli$^2$ \\ {\tt gergatsouli@wisc.edu} \and
\\
Charilaos Pipis$^1$ \\ {\tt chpipis@softlab.ntua.gr} \and
\\
Miltiadis Stouras$^{3}$ \\ {\tt miltiadis.stouras@epfl.ch} \and
\\
Christos Tzamos$^2$ \\ {\tt tzamos@wisc.edu}
\\
\\
$^1$National Technical University of Athens\\
$^2$University of Wisconsin-Madison\\
$^3$ École Polytechnique Fédérale de Lausanne\\
}
\date{}

\begin{document}

\maketitle
\begin{abstract}

Graph connectivity is a fundamental combinatorial optimization problem that arises in many practical applications, where usually a spanning subgraph of a network is used for its operation. However, in the real world, links may fail unexpectedly deeming the networks non-operational, while checking whether a link is damaged is costly and possibly erroneous. After an event that has damaged an arbitrary subset of the edges, the network operator must find a spanning tree of the network using non-damaged edges by making as few checks as possible.

Motivated by such questions, we study the problem of finding a spanning tree in a network, when we only have access to noisy queries of the form “Does edge e exist?”.  We design efficient algorithms, even when edges fail adversarially, for all possible error regimes;  2-sided error (where any answer might be erroneous), false positives (where “no” answers are always correct) and false negatives (where “yes” answers are always correct). In the first two regimes we provide efficient algorithms and give matching lower bounds for general graphs. In the False Negative case we design efficient algorithms for large interesting families of graphs (e.g. bounded treewidth, sparse).
Using the previous results, we provide tight algorithms for the practically useful family of planar graphs in all error regimes.

\end{abstract}

\setcounter{page}{0}
\thispagestyle{empty}
\newpage

\section{Introduction}

From road and railway networks to electric circuits and computer networks, maintaining an operational spanning subgraph of a network is crucial in many real-life applications. This problem can be viewed as a traditional combinatorial optimization problem, graph connectivity. However, in real-life applications, edges sometimes fail; a natural destruction on the road or a malfunction in the connector cable can render a link in the network non operational. Worse still, detecting faulty connections is not always that simple; our tools may be unreliable, resulting in untrustworthy indications. Can we efficiently find an operational spanning tree, when we cannot obtain reliable signals on the operation of edges?

\textbf{In this work we study the problem of finding a spanning tree using noisy queries on the existence of edges}. Specifically, given a moldgraph $G$, where some subgraph $H$ of $G$ is realized, we have access to an oracle that answers questions of the form \emph{``Does edge $e$ exist in $H$?"}. The answers the oracle gives are \emph{inconsistent}: for a specific edge $e$, the answer might differ every time. Queries to the oracle are costly, therefore our goal is to find a spanning tree asking as few queries as possible. We design efficient algorithms that achieve this goal, in all 3 different error regimes; 2-sided (where any answer might be erroneous), 1-sided false positives (where “no” answers are always correct) and 1-sided false negatives (where “yes”
answers are always correct).

\subsection{Our Contribution}
As a warm-up, we begin by solving the simpler problem of verifying the connectivity of a tree (Theorem~\ref{thm:verify-connectivity-linear}). In the 2-sided error case, we give an algorithm performing $O(m\log m)$ queries and show that no algorithm can perform better, even on special cases like sparse graphs. In the 1-sided False Negative error regime, we design an algorithm that is optimal on planar graphs and yields
efficient guarantees for other special families like graphs with treewidth $k$ or degeneracy $k$ ($O(km)$ queries) and graphs with Hadwiger number $k$ ($O(k \sqrt{\log k} m)$ queries). The same algorithm
can be used for general graphs, and outside the aforementioned
families, its performance gracefully degrades to that of the naive
strategy (that is $O(m \log n)$ queries).

In the False-Positive error case, our algorithm obtains tight guarantees for general graphs~(Theorem~\ref{thm:combined-fp-algo}), while in the special case where the realized graph is acyclic, the query complexity becomes linear.
Both in the False-Negatives and the False-Positives case, our algorithms \emph{do not need to know} whether the graph has any special properties (sparsity, acyclic realized subgraph etc), they can simply run a unified algorithm that
achieves the best guarantee according to each case. 

Our results, imply tight algorithms for planar graphs, which is a family of graphs that is frequently encountered in road/railway networks, electrical circuits, image processing/computer vision~\cite{YarkFowl2015,SchmToppCreme2009}.

\input{related_work}

\section{Preliminaries}\label{sec:prelims}
\input{model}

\section{Verifying Connectivity}\label{sec:verify}
\input{verifying_connectivity}

\section{Two-sided error oracle}\label{sec:two_sided}
\input{2-sided-error}

\section{One-sided error oracle: False Negatives}\label{sec:FN}
\input{FN}

\section{One-sided error oracle: False Positives}\label{sec:FP}
\input{FP}

\bibliographystyle{alpha}
\bibliography{bib}
\newpage
\appendix
\input{appendix}

\end{document}

%% file: related_work.tex
\subsection{Related Work}
Variations of this problem have been studied in literature, however they differ from our setting. A crucial assumption in all of the previous works, is that each edge has an independent probability of existence, which
however, does not always bear out in practice. For example, a power outage
may deem several local network links non-operational. In our work,
we allow the existence probabilities of different edges to be arbitrarily
correlated.

\cite{upfal_noisy_94} first studied evaluating boolean decision trees where the nodes are noisy and the goal is to find the correct leaf within a tolerance parameter $Q$. In \cite{FuWangKuma2014,FuFuXuPengWangLu2017} the authors study finding spanning trees in Erdos-Renyi graphs, a special class of random graphs where each edge exists independently with some probability. In their setting each edge has a \emph{query cost}, and the goal is to find a spanning tree using the least number of queries. However, contrary to their setting, we handle both adversarially selected realized graphs and noisy answers to existence queries.
More recently, \cite{LyuRenAbbaZhan2021} also studied the problem of finding a MST when the edges can fail with some probability. The work of \cite{HoffErleKrizMihalRama2008} and \cite{ErleHoff014} considers verifying spanning trees where each edge has a weight inside an uncertainty region and the algorithm can ask the edge to reveal it. The goal is to find a small weight spanning tree without asking all the edges for their weight. None of these settings account for \emph{inconsistent} queries, like in our setting. 

On a different problem in \cite{mazumdar_clustering_2017}, the authors study clustering in a similar setting to ours where they have access to queries of the  form ``is $u$ and $v$ in the same cluster?". However in their model, the oracle gives consistent answers on each query\footnote{In our case, asking if an edge $e$ exists can change from ``Yes" to ``No" or vice versa for distinct queries on the same edge.}.

%% file: model.tex
In \noisygraph{} we are given a graph $G=(V,E)$, 
with edge set $E$ and node set $V$ called the \emph{moldgraph}, then an adversary selects a \textit{connected} subgraph $G'$ of $G$ to be realized.
The goal of the algorithm is to find a spanning tree in the subgraph $G'$ spending as little time as possible gathering information. The algorithm does not directly observe the subgraph that is realized, but only has access to an oracle that answers questions of the form ``\emph{Is edge $e\in E$ realized?}". Each call to this oracle costs $1$, therefore our goal is to find a realized spanning tree, with constant probability using the minimum number of queries to the oracle. 

This oracle, however, is noisy and inconsistent; it might not give the correct answer and when asked multiple times on the same edge, it may give different answers. More formally, the oracle is a function $\oracle: E\rightarrow \{\text{Yes, No}\}$, that given an edge $e\in E$ answers either \emph{Yes} or \emph{No}, indicating whether the edge is realized or not. We study all the possible error types for the oracle, outlined below.

\paragraph{2-sided error:} the oracle's answer is wrong  (ie. with ``No" for realized edges and ``Yes" for non-realized) with \textit{constant} error probability $p < 1/2$. 

\paragraph{1-sided error, False Negative (FN):} if the edge is not realized then the response is always ``No". If it is realized then the response is ``No" with \textit{constant} error probability $p < 1/2$, and ``Yes" with probability $(1-p)$. Thus, when the oracle responds ``Yes" then it is certain that this edge is realized, but when it responds ``No" then the edge may or may not be realized, hence the False Negative responses.

\paragraph{1-sided error, False Positive (FP):} if the edge is realized then the response is always ``Yes". If it is not realized then the response is ``Yes" with \textit{constant} error probability $p < 1/2$, and ``No" with probability $(1-p)$. Thus, when the oracle responds ``No" then it is certain that this edge is not realized, but when it responds ``Yes" then the edge may or may not be realized, hence the False Positive responses.

\subsection{Graphs notation}

We present some useful definitions for concepts we use throughout the paper. We begin by our definition of sparsity. Sometimes the parameter $\rho$ is referred to as the average degree of the graph.

\begin{definition}[$\rho$-sparse Graphs]
A graph $G$ with $m$ edges and $n$ vertices is called $\rho$-sparse if $m \leq \rho n$. 
\end{definition}

It is noteworthy that our definition of $\rho$-sparsity for graphs is weaker than the usual definitions of graph sparsity, in the sense that it can be easily satisfied without requiring local sparsity properties. This means that constructing efficient algorithms for this property, directly translates to efficient algorithms for other usual sparsity parameters. For example, using our definition, $k$-degenerate graphs are $k$-sparse \cite{lick_white_1970}, $k$-treewidth graphs are also $k$-sparse \cite{robertson1984graph}, and graphs with Hadwiger number $k$ are $O(k \sqrt{\log k})$-sparse \cite{kostochka1984lower, bollobas1980hadwiger}.

\begin{definition}[Edge Contraction]
The edge contraction operation on an edge $e = \{u, v\}$ of a graph $G$ results in a new graph $G'$ wherein $u$ and $v$ are replaced by a new vertex $uv$ which is connected to all vertices of $G$ that were incident to either $u$ or $v$.
\end{definition}

\begin{definition}[Graph minor~\cite{lovasz2006}]
A graph $H$ is called a minor of a graph $G$, if $H$ can occur after applying a series of edge deletions, vertex deletions, and edge contractions.
\end{definition}

\begin{definition}[Minor-closed Graph Family]
A set $\mathcal{F}$ of graphs is called a minor-closed graph family if for any $G \in \mathcal{F}$ all minors of $G$ also belong to $\mathcal{F}$.
\end{definition}

The algorithms that we present for FN and FP queries produce and handle multigraphs during intermediate steps. Here, we present and define some of the main concepts we are using from multigraphs. 

An edge $e = \{u, v\}$ in a multigraph is a link connecting $u$ and $v$. There might be other parallel edges between $u$ and $v$, and each one is distinct. Sometimes we will need to work with the set of all parallel edges between $u$ and $v$. We call this set a \emph{super-edge} between $u$ and $v$. For simplicity, we consider simple edges to be super-edges with size 1.

\begin{definition}[Neighborhood]
    For a vertex $v\in V$, we denote by $N(v)$ its \emph{neighborhood}, which is the set of all super-edges with $v$ as one of the endpoints.
\end{definition}

\begin{definition}[Degree]
For a node $v\in V$, degree $\text{deg}(v)$ of $v$ is the size of $N(v)$.
\end{definition}

Finally, we note that whenever an algorithm performs an edge-contraction operation on edge $e$ of a multigraph, then we delete all the other parallel edges to $e$ as well, leaving no self-loops in the graph. Moreover, if any parallel super-edges result after the contraction, they are replaced by a larger super-edge, their union.

%% file: verifying_connectivity.tex
As a warm-up, before attempting to find a realized spanning tree of the graph,
we begin by presenting an algorithm for verifying the connectivity of a tree.
Specifically, given a tree, which could be the operating spanning tree of a
network, we want to verify whether it is connected or not, that is whether 
all of its edges are realized. 

Naively verifying this property (ie. performing a fixed number of queries 
on every edge) needs $O(n\log n)$ queries to yield a constant probability
guarantee. However, we show that it is possible to verify the connectivity
of the tree using only $O(n)$ queries, which is asymptotically the same
as if our oracle had no noise.

\begin{algorithm}[tb]
\caption{$Verify(T)$: A verification protocol for tree connectivity with 2-sided error}
\textbf{Input}: Tree $\mathcal{T}$ of $n$ edges\\
\textbf{Parameters}: $\e$, $\delta$, $p$\\
\textbf{Output}: True iff $T$ is connected
\label{alg:verify-connectivity}
\begin{algorithmic}[1] 
    \STATE $threshold\ \gets \lceil \log_{\frac{1-p}{p}} \left( \frac{1}{\delta} \right) \rceil$ 
    \STATE $budget\ \gets \lceil \frac{1}{\epsilon} \cdot \frac{1}{1-2p} \rceil \cdot threshold \cdot n$ 
    \FOR{$e \in E(T)$}
        \STATE counter $\gets 0$ 
        \WHILE{(counter $<$ threshold) \& (budget $>$ 0)}
            \STATE q $\gets$ Query edge $e$ 
            \STATE budget $\gets$ budget $-1$ \\
             \IF{q is ``Yes''}
                \STATE counter $\gets$ counter $+1$ 
            \ELSE
                \STATE counter $\gets$ counter $-1$
            \ENDIF
        \ENDWHILE
        \IF{budget = 0}
            \STATE \textbf{return} False
        \ENDIF\\
    \ENDFOR\\
    \STATE \textbf{return} True
\end{algorithmic}
\end{algorithm}

Our algorithm (Algorithm~\ref{alg:verify-connectivity}) has a
predetermined budget of queries that depends on the requested guarantees.
It fixes an ordering of the edges, and for each edge it performs queries 
until it receives $c$ more positive answers than negative ones (where $c$ is
also fixed and depends on the requested guarantees). For realized
edges, since $p < 1/2$, the expected number of queries until they reach $c$
is polynomial in $c$, while the probability to never yield $c$ more 
positive answers than negative ones, within the budget, is exponentially 
small in $c$. Using a global budget instead of a fixed per-edge number of queries, allows us to dynamically allocate it based on need.
We save up queries from realized edges that reach the
threshold quickly and spend it on other realized edges that yield multiple
false negatives. On the other hand, the probability that a non-realized edge
reaches $c$ is exponentially small in $c$. As a result, if a non-realized
edge exists in the tree, the algorithm will consume all of its budget on this
edge before reaching the threshold. Using this approach, 
Algorithm~\ref{alg:verify-connectivity} achieves the following guarantees.

\begin{restatable}{theorem}{verifyCon}\label{thm:verify-connectivity-linear}
Algorithm~\ref{alg:verify-connectivity} correctly classifies connected trees
with probability $1 - \e$ and disconnected trees
with probability $1 - \delta$, while performing at most
$O\left( \frac{1}{\e} n \log \frac{1}{\delta} \right)$ queries.
\end{restatable}
The proof of the theorem is deferred to
Section~\ref{appendix:sec_verify} of the appendix. 
As a Corollary of Theorem~\ref{thm:verify-connectivity-linear}, we get
that, for constant probability guarantees, Algorithm~\ref{alg:verify-connectivity}
requires $O(n)$ queries.

\begin{corollary}
    Algorithm~\ref{alg:verify-connectivity} correctly classifies connected
trees with probability $.99$ and disconnected trees with probability $.99$, using $O(n)$ queries.
\end{corollary}

%% file: 2-sided-error.tex
Moving on to our main result, we show how to find a realized spanning tree in the 2-sided error oracle. Recall that in this case, the oracle may give false responses either when the edge is realized or when it is non-realized. The main result for this regime is the following theorem.

\begin{theorem}
 In the 2-sided error regime, there exists an algorithm that finds a realized spanning tree with high probability, using $O(m \log m)$ queries in a moldgraph of $m$ edges. 
 Moreover, no algorithm can do better than $\Omega(m \log m)$.
 \end{theorem}

To show this theorem, we separately show the upper and lower bounds in Lemmas~\ref{thm:2-sided-error-central} and \ref{lem:2-sided-error-lb}. Combining them we immediately get the theorem.

\subsection{Upper Bound}
In order to show the upper bound (Lemma~\ref{thm:2-sided-error-central}), we 
describe an algorithm that achieves this query bound and prove its correctness
in Lemma~\ref{lem:naive-two-sided-high-prob}

\begin{lemma}\label{thm:2-sided-error-central}
In the 2-sided error regime, there exists an algorithm that finds a realized spanning tree with high probability, using $O(m \log m)$ queries in a moldgraph of $m$ edges.
\end{lemma}

Our algorithm, described below, uses the same idea as the naive algorithm for the connectivity verification problem.

\paragraph{Naive algorithm} The algorithm is as follows. First, query each edge of the moldgraph $O(\log{m})$ times. Second, by treating all edges with more ``Yes" than ``No" responses as realized, compute a spanning tree on the discovered realized subgraph. If the graph is disconnected, output any spanning tree of the moldgraph.

The algorithm performs $O(m \log m)$ queries in total. The following lemma shows the correctness of the algorithm, and its proof is deferred to Section~\ref{appendix:sec_2sided} of the appendix. 

\begin{restatable}{lemma}{naiveTwoSided}\label{lem:naive-two-sided-high-prob}
The Naive algorithm finds a realized spanning tree with high probability.
\end{restatable}

\begin{figure}
\centering
\begin{subfigure}{\linewidth}
    \centering
    \includegraphics[width=0.73\linewidth]{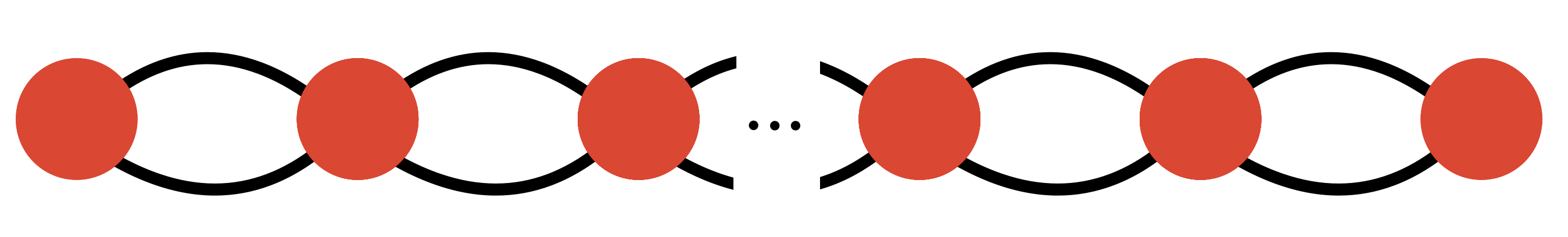}
    \caption{The moldgraph used for the proof of Lemma~\ref{lem:2-sided-error-lb}}
    \label{subfig:2sided-moldgraph}
\end{subfigure}%
\newline
\begin{subfigure}{\linewidth}
    \centering
    \includegraphics[width=0.73\linewidth]{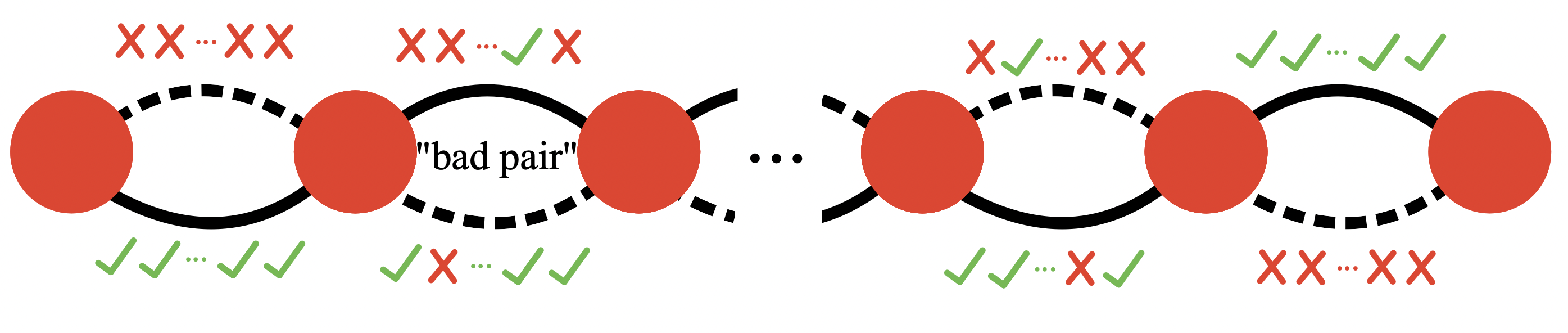}
    \caption{A realized graph from the family mentioned in the proof of
    Lemma~\ref{lem:2-sided-error-lb}. The dotted lines are non-realized edges
    and by the edges lie the answers of the 2-sided noisy oracle.}
    \label{subfig:2sided-LB}
\end{subfigure}
\caption{}
\label{fig:2sided-LB}
\end{figure}

\subsection{Lower Bound}

Switching gears towards the lower bound, we give an instance showing that we cannot hope for anything better than the $O(m\log m)$ upper bound shown in the previous section. As we expected, this is a strictly
harder problem than verifying the connectivity of a tree, since
in this case we have to identify a realized edge in each independent
cut of the moldgraph.

\begin{lemma}\label{lem:2-sided-error-lb}
In the 2-sided error regime, there exists a graph where any algorithm requires $\Omega (m \log m)$ queries to discover a spanning tree with constant failure probability.
\end{lemma}

\begin{proof}[Proof of Lemma~\ref{lem:2-sided-error-lb}]
We consider the graph in Figure~\ref{fig:2sided-LB} which has $n+1$ vertices and $2 n$ edges. Between any two consecutive vertices there exist 2 parallel edges, one of which is realized and the other is not. Note that this moldgraph is a multigraph for simplicity of presentation. But the arguments hold analogously for a $2$-by-$n$ grid moldgraph, if we additionally give our algorithms the extra information that all vertical edges are realized. This only makes an algorithm more powerful as, in the worst-case, it can choose to ignore this information and execute independently. Then, we show that even with this extra knowledge, no algorithm can achieve a better query complexity than the desired bound.

Any algorithm on this graph will have to treat each edge-pair independently of the others, because uncovering a realized edge in one pair does not give any information on the solution to other pairs. Moreover, the optimal strategy for any algorithm is to query both edges in a pair an equal number of times and then pick the one with the majority of positive responses. Otherwise, picking the one with the least positive responses gives us smaller probability of success.

Let us suppose that the algorithm performs $k$ queries on a specific edge-pair. If more than $k/2$ of them give false responses then the wrong edge will be picked as realized. Consequently, the algorithm will have to perform more than $k$ queries on this edge-pair to discover the correct edge. We set $k = \Theta(\log (n / \log n) ) = \Theta(\log n)$ and call such a faulty edge-pair a ``bad" pair. Then, by the fact that false queries follow a binomial distribution, we get

$$\Pr{}{\text{bad}} > p^{k/2} > \log n / n$$

\noindent Thus the probability that \textit{at least} one such ``bad pair" exists is
\begin{align*}
\Pr{}{\text{bad exists}} & = 1 - \Pr{}{\text{all good}} \\ 
& \geq 1 - \left( 1 - \frac{\log n}{n} \right)^n \\
& \geq 1 - \frac{1}{n} \\
\end{align*}
Assuming that there exists such a ``bad pair" then the algorithm will have to perform $\Omega(\log {n})$ queries on it. However, there is no way to know from the beginning which pair it will be. This means that any algorithm in this case will have to pick an $\alpha$-fraction of edge-pairs on which to perform these many queries, for some (not necessarily constant) $\alpha$. Define $Q_b$ to be the event that the algorithm performs $\Omega(\log n)$ queries on the ``bad pair". Then, $\Pr{}{Q_b} = \alpha$ and the probability of failure for the algorithm when a bad pair exists is
$$
\Pr{}{\text{fail} | \text{bad exists}} = \alpha \Pr{}{\text{fail} | Q_b} + (1 - \alpha) \geq (1 - \alpha)
$$

\noindent Thus, the total probability of failure to discover a spanning tree will be

\begin{align*}
    \Pr{}{\text{fail}} &\geq \Pr{}{\text{fail} | \text{bad exists}} \Pr{}{\text{bad exists}}\\
    &\geq (1 - \alpha) \left( 1 - \frac{1}{n} \right)
\end{align*}

If we want the failure probability to be constant then it must hold that $\alpha = 1 - o(1)$. Consequently, we perform $\Omega(\log n)$ queries on $\alpha n = \Theta(n)$ edges, giving us an $\Omega(n \log n)$ lower bound for this graph instance. The desired lower bound comes from the fact that any algorithm performing better than $\Omega(m \log m)$ queries would perform better than $\Omega(n \log n)$ in this instance, as it has $m = O(n)$.
\end{proof}

%% file: FN.tex
Moving on to the one-sided error regime, we first consider the case of False Negative errors in the oracle.
Our main contribution is an algorithm for the case of sparse and minor-closed graph families. We also show how we can obtain an algorithm for general graphs, that still performs well on $\rho$-sparse graphs, even when this property is unknown to us.

Finally, we show how to combine these results to obtain an algorithm for planar graphs that uses $O(m)$ queries to find a spanning tree, while no algorithm can do better (Corollary~\ref{cor:planar-fn}). This also implies better guarantees for other special families of graphs, as shown in Corollary~\ref{cor:other-special-fn}.

\begin{theorem}\label{thm:combined-fn-algo}
In the 1-sided False Negative regime, given a moldgraph of $m$ edges and $n$ nodes, there exists an algorithm that finds a realized spanning tree using $O(m\log n)$ queries in expectation. If the moldgraph is $\rho$-sparse and minor-closed, the algorithm uses $O(\rho m)$ queries in expectation.
\end{theorem}

\begin{proof}[Proof of Theorem~\ref{thm:combined-fn-algo}]
To create such an algorithm, we can run in parallel both the Naive Algorithm presented in the next subsection (Lemma~\ref{lem:naive_FN}) and Algorithm~\ref{alg:sparse-fn-algo} for sparse minor-closed graphs. In particular, it suffices to let the two algorithms perform their queries alternately and independently of each other.

If the moldgraph is $\rho$-sparse and minor-closed, then by Lemma~\ref{lem:sparse-fn-algo}, Algorithm~\ref{alg:sparse-fn-algo} will perform $O(\rho m)$ queries in expectation before finding a spanning tree. For each one of these queries, we have just performed one extra query for the naive algorithm. Thus, eventually the combined algorithm will find the spanning tree in the same query complexity.

Otherwise, in general graphs, by using Lemma~\ref{lem:naive_FN} for the Naive Algorithm and a similar argument as before, we get $O(m \log n)$ queries in expectation.
\end{proof}

\subsection{Warm-up: naive algorithm}
We start by describing the naive algorithm used for general graphs, that obtains the $O(m\log n)$ guarantee. The algorithm proceeds in rounds, where in each round it performs one query on all $m$ edges of the moldgraph. It repeats this process until $n-1$ realized edges forming a spanning tree are discovered. Note that in the case of FN queries we can be completely certain when a realized edge is discovered.

\begin{restatable}{lemma}{naiveFN}\label{lem:naive_FN}
The Naive Algorithm finds a spanning tree while performing $O(m\log n)$ queries
on expectation. 
\end{restatable}
The proof is deferred to section~\ref{appendix:sec_fn} of the Appendix.

\subsection{An Algorithm for $\rho$-sparse graphs}
We present the algorithm used to obtain better guarantees for the family of $\rho$-sparse and minor closed graphs. The algorithm is presented in Algorithm~\ref{alg:sparse-fn-algo} and the guarantees in Lemma~\ref{lem:sparse-fn-algo}.

\begin{lemma}\label{lem:sparse-fn-algo}
Given a moldgraph of $m$ edges, belonging to a $\rho$-sparse and 
minor-closed family of graphs, Algorithm~\ref{alg:sparse-fn-algo} performs 
$O(\rho m)$ queries in expectation and uncovers a realized spanning tree with 
probability $1$.
\end{lemma}

The basic primitive that our algorithm uses is a \texttt{DISCOVER} subroutine that cleverly orders queries in a way that it allows it to explore multiple edge sets in parallel. The subroutine is presented in more detail in Algorithm~\ref{alg:parallel-discovery} and the Lemma that follows.

\centerline{}

\centerline{}

\begin{algorithm}[tb]
\caption{DISCOVER($\mathcal{S}$): Discovers an edge 
in a collection $\mathcal{S}$ of edge sets.}
\label{alg:parallel-discovery}
\textbf{Input}: Set $\mathcal{S} = \{E_1, E_2, \ldots, E_k\}$\\
\textbf{Output}: Realized edge $e$
\begin{algorithmic}[1]
        \WHILE{no edge is found}
            \FOR{$i = 1$ \textbf{to} $k$}
                \STATE Let $e$ be the next edge in the cyclic order of $E_i$
                \STATE $q \gets$ Query edge $e$ \\
                \IF{$q$ is ``Yes"}
                    \STATE \textbf{return} $e$
                \ENDIF
            \ENDFOR
        \ENDWHILE
\end{algorithmic}
\end{algorithm}
\centerline{}

\begin{restatable}{lemma}{parallelDiscovery}\label{lemma:parallel-discovery}
Let $\mathcal{S}$ be a collection of $k$ sets of edges. Assume as well that there exists at least one set in $\mathcal{S}$ containing some realized edge, then Algorithm~\ref{alg:parallel-discovery} finds and returns such an edge using at most $2k|E_f|$ queries in expectation, where $E_f$ is the set containing the found edge.
\end{restatable}

The proof is deferred to Section~\ref{appendix:sec_fn} of the Appendix.

\centerline{}

\begin{algorithm}[tb]
\caption{SolveSparseFN($G$): Solves the problem in sparse moldgraphs using an 1-sided error oracle with False Negatives.}
\label{alg:sparse-fn-algo}
\textbf{Input}: Graph $G$\\
\textbf{Output}: A realized spanning tree
\begin{algorithmic}[1]
    \IF{$G$ is a single node}
        \STATE \textbf{return} $\emptyset$
    \ENDIF\\
    \STATE $u = \text{argmin} \{\text{deg}(v)$\}
    \STATE $e = \text{DISCOVER}(N(u))$\\
    \STATE // Let $G'$ be the resulting graph after contracting edge $e$.
    \STATE \textbf{return} $\text{SolveSparseFN}(G')\ \cup\ \{e\}$
\end{algorithmic}
\end{algorithm}
\centerline{}

We are now ready to prove the main result of the section, 
Lemma~\ref{lem:sparse-fn-algo}.

\begin{proof}[Proof of Lemma~\ref{lem:sparse-fn-algo}]
We use induction on the number $n$ of vertices of $G$ to prove that Algorithm~\ref{alg:sparse-fn-algo} performs at most $4 \rho m$ queries in expectation. Throughout this proof we denote by $m$ the total number of simple edges of the graph, and by $m_s$ the number of super-edges.

The base case of induction for $n = 1$ holds trivially, as the algorithm needs no queries. Now assume that it holds for all graphs of at most $n - 1$ vertices. We will prove that it also holds for graphs of $n$ vertices.

By using the sparsity property that $m_s \leq \rho n$, we can see that there always exists at least one vertex $u$ with degree $\text{deg}(u) \leq 2 \rho$. Otherwise, the number of super-edges would be equal to
$$m_s = \frac{1}{2} \sum_{v \in V} \text{deg}(v) > \rho n$$

which leads to the graph not being $\rho$-sparse, a contradiction.

The algorithm uses Algorithm~\ref{alg:parallel-discovery} as a basic subroutine, passing as input the neighborhood $N(u)$ of the least-degree vertex $u$. This is allowed because $N(u)$ is a cut of the moldgraph and thus satisfies the assumption of Lemma~\ref{lemma:parallel-discovery} that at least one edge-set contains a realized edge.

The \texttt{DISCOVER} subroutine will find a realized edge inside a
super-edge $E_f$, spending in expectation at most $2 \cdot \text{deg}(u) |E_f| \leq 4 \rho |E_f|$ queries.

After contracting the super-edge $E_f$ that contains
the realized edge found by the subroutine, the algorithm
will continue on the contracted graph $G'$ that has $n-1$ vertices. It is important to note here that the graph $G'$ will still be $\rho$-sparse, because of our assumption that graphs are also minor-closed. By the induction hypothesis, the algorithm will find a realized spanning tree of $G'$, by performing at most

$$
4 \rho\ (m - |E_f|)
$$

\noindent queries in expectation. Therefore, the total number of queries performed will be at most

$$
4 \rho |E_f| + 4 \rho\ (m - |E_f|) = 4 \rho m
$$

\noindent in expectation, and thus gives us the desired $O(\rho m)$ upper bound.
\end{proof}

\begin{corollary}\label{cor:planar-fn}
If the moldgraph is planar, then Algorithm~\ref{alg:sparse-fn-algo} performs $O(m)$ queries in expectation and succeeds in finding a realized spanning tree with probability 1. Moreover, this is tight and no better algorithm exists for this case.
\end{corollary}

\begin{proof}[Proof of Corollary~\ref{cor:planar-fn}]
The upper bound follows directly by Lemma~\ref{lem:sparse-fn-algo} and the fact that planar graphs are both minor-closed and $3$-sparse, as $m \leq 3 n - 6$.

To argue that the algorithm is tight, we remind that in all cases the realized spanning subgraph is chosen arbitrarily, and even in an adversarial way. This means that any algorithm aiming to find a realized spanning tree has to perform at least $n-1$ queries, for the $n-1$ realized edges of the spanning tree that it discovers. As planar graphs have $m = \Theta(n)$ the desired lower bound follows.
\end{proof}

Apart from planar graphs, as we hinted when defining $\rho$-sparsity, an efficient algorithm for $\rho$-sparse graphs immediately gives us efficient algorithms for graphs that are sparse with respect to other widely-used parameters. Based on \cite{lick_white_1970, robertson1984graph,  kostochka1984lower}, we can see that the parameters of degeneracy, treewidth and Hadwiger Number all define minor-closed graph families. Moreover, $k$-degenerate graphs and $k$-treewidth graphs have sparsity $k$, while graphs with Hadwiger number $k$ are $O(k \sqrt{\log k})$-sparse. Hence, applying Lemma~\ref{lem:sparse-fn-algo} we get the following Corollary.
\centerline{}

\centerline{}

\begin{corollary}\label{cor:other-special-fn}
Algorithm~\ref{alg:sparse-fn-algo} performs $O(km)$ queries for graphs with treewidth $k$ or degeneracy $k$ and $O(k \sqrt{\log k} m)$ queries for graphs with Hadwiger number $k$.
\end{corollary}

%% file: FP.tex
In the 1-sided, False-Positive error case, we initially present an $O(m \log m)$ algorithm for general graphs, while dropping the complexity to linear when the realized graph is a tree. Each of these guarantees are tight for their respective case.
 We also show how using one algorithm, we can capture both cases, without actually knowing the structure of the realized graph, and obtain the tight guarantees for each case. Our formal guarantees are formally stated in Theorem~\ref{thm:combined-fp-algo}.

\begin{theorem}\label{thm:combined-fp-algo}
In the 1-sided error FP regime, there exists an algorithm that performs $O(m \log m)$ queries in the general case and finds a realized spanning tree with high probability. Moreover, if the realized subgraph is a tree, the algorithm performs $O(m)$ queries in expectation and succeeds to find the spanning tree with high probability. These complexities are tight for their respective cases.
\end{theorem}

\begin{proof}[Proof of Theorem~\ref{thm:combined-fp-algo}]
Similarly to the case of False Negative oracle errors (Lemma~\ref{thm:combined-fn-algo}) we can combine the Naive Algortihm with Algorithm~\ref{alg:solve-planar-fp-dual}, letting them perform queries alternately. When any of the two algorithms halts, the process ends.

Consider the case of general graphs. By Lemma~\ref{lem:naive_FP} the naive algorithm succeeds in finding a realized spanning tree with high probability by performing $O(m \log m)$ queries. Thus, the combined algorithm also succeeds with at most $O(m \log m)$ queries.

Now, we look at the special case where the moldgraph is planar and the realized subgraph is a tree. By Lemma~\ref{lem:dual-alg-fp-tree}, Algorithm~\ref{alg:solve-planar-fp-dual} performs at most $O(m)$ queries in expectation. Let $X$ be the random variable denoting the number of queries performed by Algorithm~\ref{alg:solve-planar-fp-dual}, and let $E = \{X > m \log_2(m^2)\}$ be the event that $X$ is more than the queries of the naive algorithm. Using Markov's inequality we get
$$
\Pr{}{E} \leq \frac{\E{X}}{m \log_2(m^2)} \leq O\left( \frac{1}{\log m} \right)
$$
Thus, the number $Q$ of queries performed before the combined algorithm terminates for this case is
\begin{align*}
    \E{Q} &= \E{Q | \neg E} \Pr{}{\neg E} + \E{Q | E} \Pr{}{E}\\
    &\leq \E{X} + (m \log_2(m^2)) O\left( \frac{1}{\log m} \right)\\
    &\leq O(m)
\end{align*}

When the algorithm has terminated, there still exists some probability of failure, in the case that both event $E$ happens and the naive algorithm fails to find a realized spanning tree. By using the failure probability from Lemma~\ref{lem:naive_FP} and the previously calculated probability of event $E$ happening, we get
$$
\Pr{}{E \wedge \text{naive algo fails}} = O\left( \frac{1}{m \log m} \right)
$$

Thus, the combined algorithm succeeds in this case with high probability as well, and the proof is concluded.
\end{proof}

\subsection{Upper \& Lower Bounds}
We start by showing, similar to the previous section, how to obtain the upper bound in the case of general graphs.
\begin{restatable}{lemma}{naiveFP}\label{lem:naive_FP}
    In the 1-sided error FP regime, there exists an
algorithm that performs $O(m \log m)$ queries in the general
case and finds a realized spanning tree with high probability.
\end{restatable}
The proof is deferred to Section~\ref{appendix:sec_fp} of the Appendix.

\noindent Furthermore, we show the lower bound for the general case graphs, stated formally below. 
Observe that from the construction of this lower bound, no algorithm can do better in planar or sparse graphs.

\begin{restatable}{lemma}{oneSidedFPLB}\label{thm:1-sided-fp-lb}
In the regime of 1-sided FP errors, if the realized graph has cycles, any algorithm requires $\Omega(m \log m)$ queries to discover a spanning tree with constant failure probability.
\end{restatable}
The proof is deferred to Section~\ref{appendix:sec_fp} of the Appendix.
\subsection{Special Case: Acyclic Realized Graphs}
In this section we focus on the special case where the realized subgraph
is a tree.
This property implies that for each cycle in the moldgraph, at least one edge of the cycle must be non-realized. Since our oracle only commits False Positive errors, the negative answers are definitive, meaning that we can completely remove
the edges which have yielded such answers. We design an algorithm
(Algorithm~\ref{alg:solve-planar-fp-dual}) that utilizes this fact
and progressively removes non-realized edges until it is left with a tree
which will also be the realized subgraph. This result is stated formally in the theorem below.

\begin{theorem}\label{thm:algo-pn-dual}
In the regime of 1-sided FP errors, if the moldgraph $G$ is a planar graph
and the realized subgraph is a tree, Algorithm~\ref{alg:solve-planar-fp-dual}
recovers the realized spanning tree with probability $1$ while performing
at most $O(m)$ queries in expectation.
\end{theorem}

Our main observation is that the FP setting is very similar, and in fact 
can be reduced to, the FN setting. In the FN setting, we know that 
at least one realized edge lies in every cut of the moldgraph and we 
repeatedly query the edges of a cut in order to unveil the realized edge. 
In a similar manner, in the FP setting,
we can repeatedly query the edges of the cycle in order to discover the non-realized edge in every cycle.
Hence, the question that arises is: how can one decide a sequence of cycles to query such that the expected number of queries is minimized?
Algorithm~\ref{alg:solve-planar-fp-dual} does so by utilizing 
the planarity of the moldgraph to consistently find small cycles to query.

\begin{algorithm}[H]
\caption{SolvePlanarFP($G$): Discovers a realized spanning tree in planar moldgraphs
using an 1-sided error oracle with False Positives.}
\label{alg:solve-planar-fp-dual}
\textbf{Input}: Planar Graph $G(V,E)$\\
\textbf{Output}: A realized spanning tree
\begin{algorithmic}[1] 
    \STATE // Construct the dual graph of $G$ 
    \STATE $G' \gets dual(G)$ \\
    \STATE // Inverse all the oracle answers and use SolveSparseFN 
    \STATE \textbf{return} $E(G) \setminus \text{SolveSparseFN}(G')$
\end{algorithmic}
\end{algorithm}

\paragraph{Description of Algorithm~\ref{alg:solve-planar-fp-dual}} 
Our algorithm reduces the FP setting of this special case to solving an 
FN instance on the \emph{dual graph} of the moldgraph. 
The dual of a planar graph has one node for each face of the planar graph, 
as illustrated in Figure~\ref{fig:planar_and_dual_graphs}.
Two face nodes are connected by an edge in the dual graph  if and only if they share an edge in the planar graph. 
With this transformation we can now focus on cuts of the dual graph which directly correspond to cycles of the moldgraph.
Notice that removing an edge in the moldgraph (after identifying it as 
non-realized) is the same as contracting it in the dual graph. As a result,
removing all the non-realized edges until we are left with a tree in the 
moldgraph is exactly equivalent to contracting the same edges in 
the dual until we are left with one node (Lemma~\ref{lem:dual-alg-fp-tree}).
For this reason, Algorithm~\ref{alg:solve-planar-fp-dual}
constructs the dual graph of the moldgraph and calls
Algorithm~\ref{alg:sparse-fn-algo} to retrieve a non-realized spanning tree 
of the dual graph.
Note that we reverse all the oracle's answers before we hand them out to 
Algorithm~\ref{alg:sparse-fn-algo}, since by design it contracts the edges for
which we get positive answers but in the FP setting we want to contract 
the edges of the dual graph whose corresponding edges yielded negative answers.

\begin{figure}[H]
\centering
\begin{subfigure}{\linewidth}
    \centering
    \includegraphics[width=0.3\linewidth]{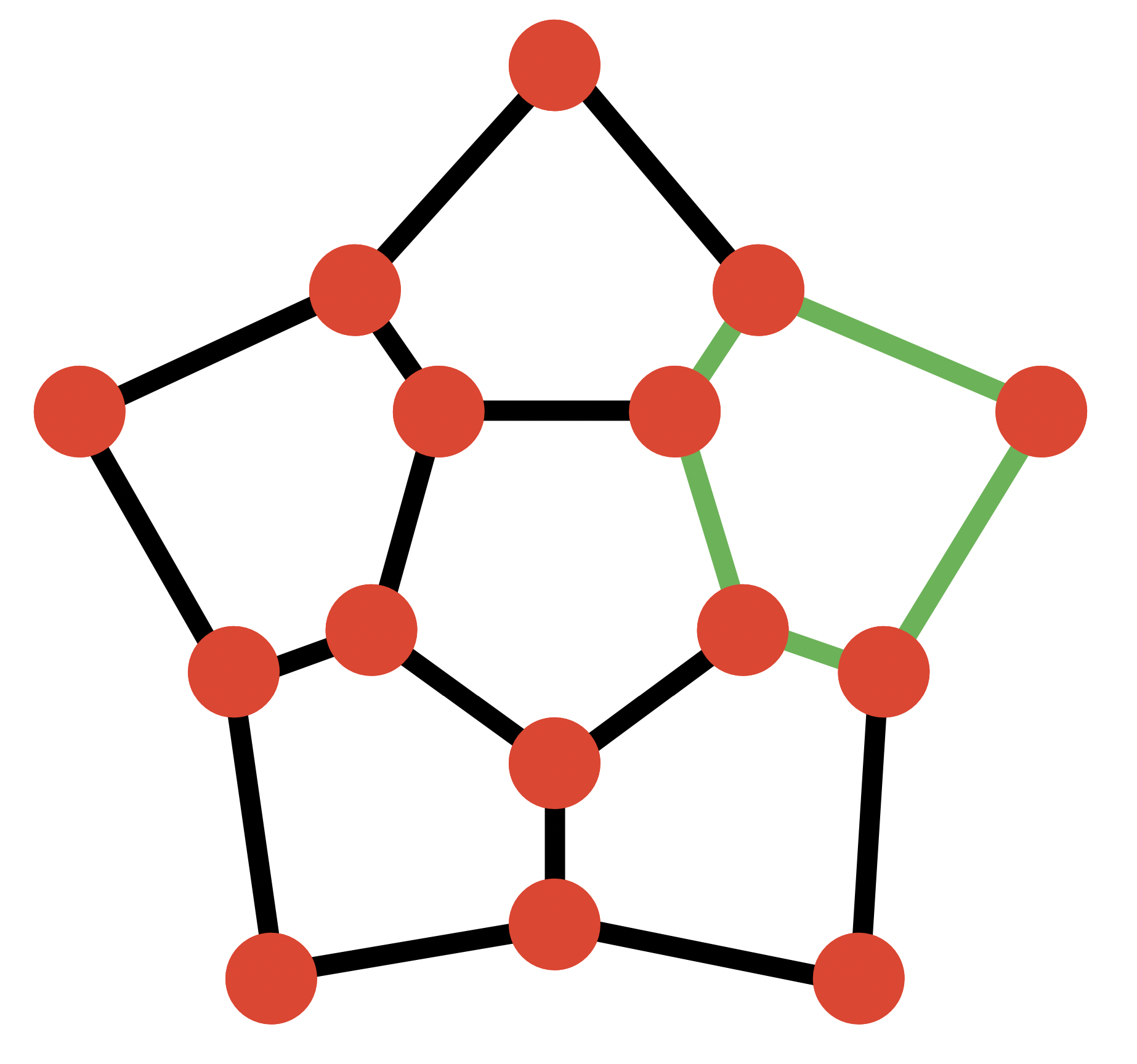}
    \hspace{0.1\linewidth}
    \includegraphics[width=0.3\linewidth]{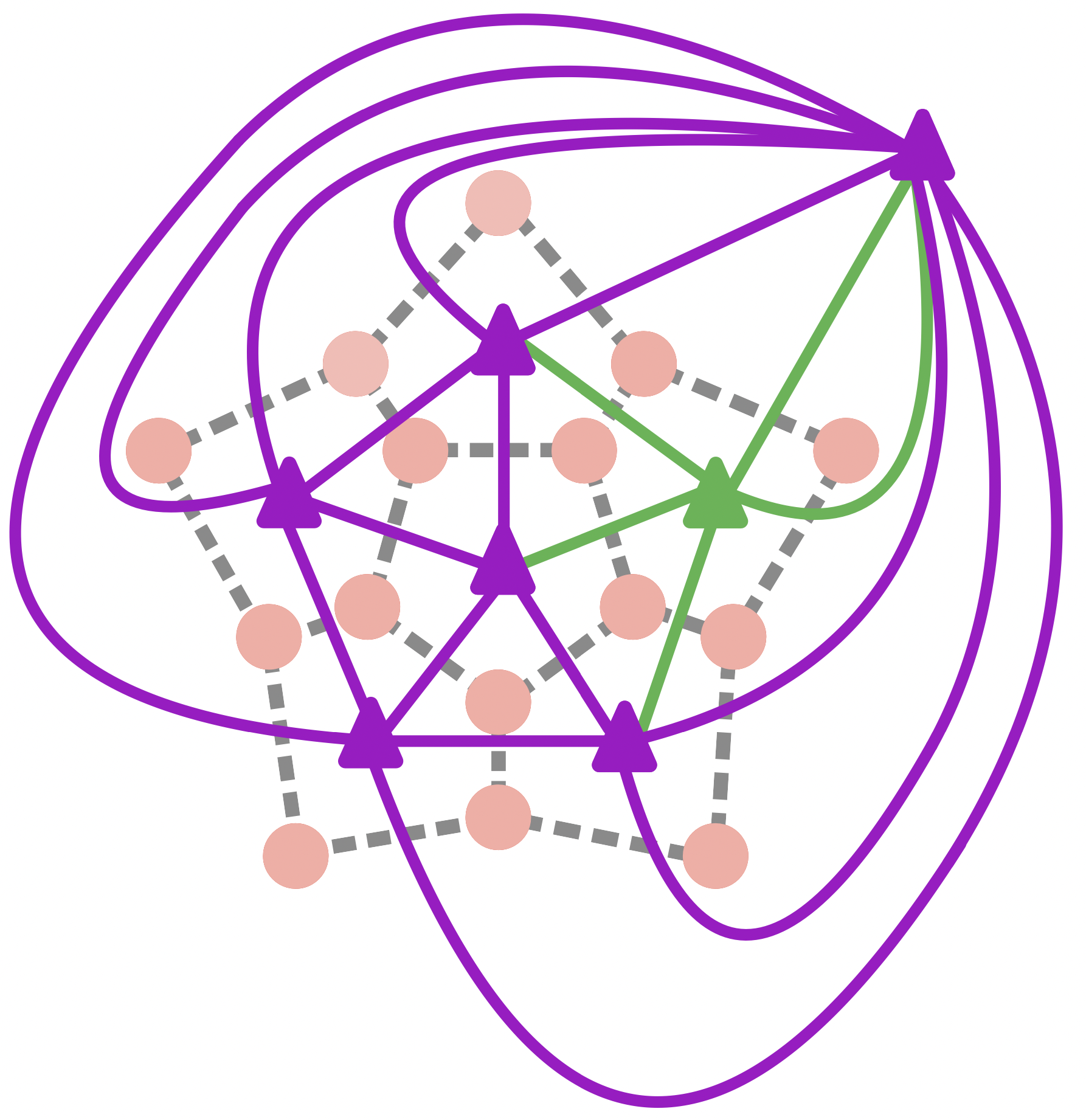}
    \caption{A cycle of the graph is the neighborhood of the enclosed face's node in the dual graph}
    \label{subfig:before_contraction}
\end{subfigure}%
\newline
\begin{subfigure}{\linewidth}
    \centering
    \includegraphics[width=0.3\linewidth]{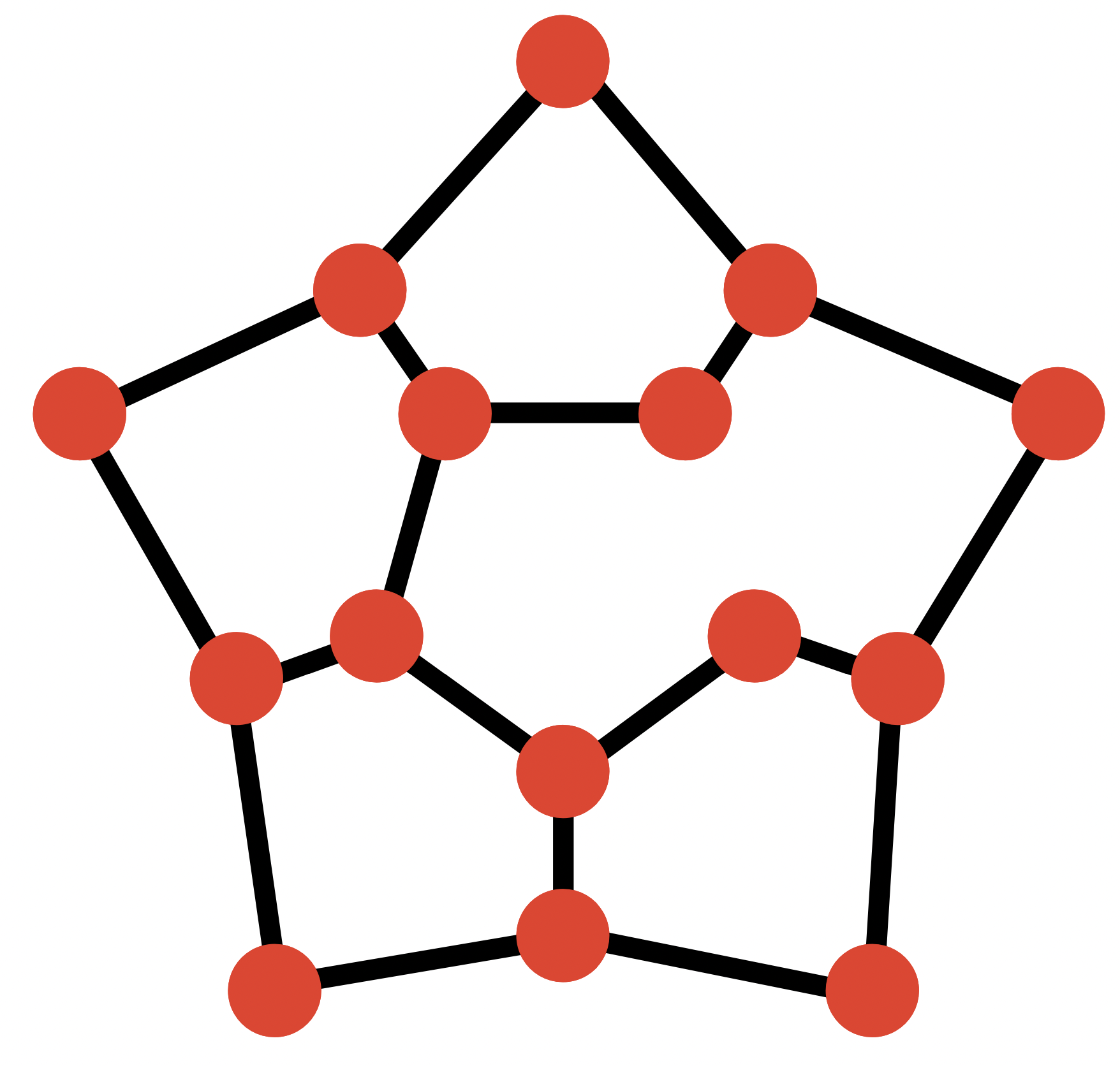}
    \hspace{0.1\linewidth}
    \includegraphics[width=0.3\linewidth]{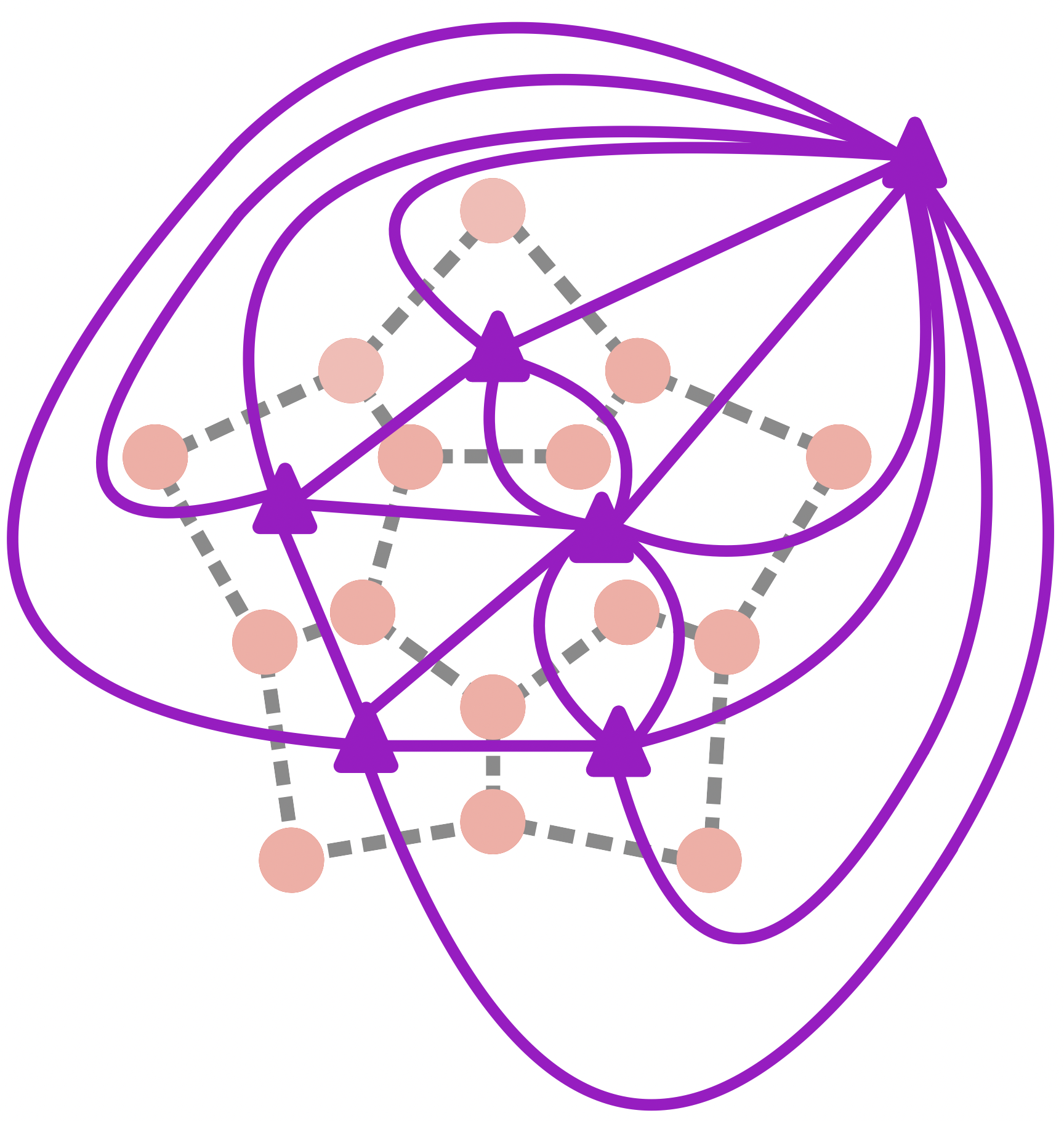}
    \caption{Breaking the cycle by removing an edge leads to the contraction of
    the corresponding edge in the dual graph}
    \label{subfig:after_contraction}
\end{subfigure}%

\caption{One iteration of Algorithm~\ref{alg:solve-planar-fp-dual}}
\label{fig:planar_and_dual_graphs}
\end{figure}

We will now prove the correctness of Algorithm~\ref{alg:solve-planar-fp-dual}
and bound its query complexity.

\begin{restatable}{lemma}{lemmaDualFP}\label{lem:dual-alg-fp-tree}
Let $G$ be a planar moldgraph, $T$ be the realized tree of $G$ and 
$G'$ be a dual graph of $G$.
The induced subgraph of $G'$ containing all the edges
corresponding to non-realized edges of $G$, is
a spanning tree of $G'$.
\end{restatable}

The proof of Lemma~\ref{lem:dual-alg-fp-tree} is deferred to 
Appendix~\ref{appendix:sec_fp}.

\begin{proof}[Proof of Theorem~\ref{thm:algo-pn-dual}]
Algorithm~\ref{alg:solve-planar-fp-dual} constructs the dual graph of $G$
and uses Algorithm~\ref{alg:sparse-fn-algo} to find a spanning tree of 
$G'$ consisting of non-realized edges. From Lemma~\ref{lem:dual-alg-fp-tree}
we know that by removing these edges, we are left with the realized tree $T$
of $G$. All the queries performed during this process are made through
Algorithm~\ref{alg:sparse-fn-algo}. Therefore, from 
Lemma~\ref{lem:sparse-fn-algo}, the total number of queries 
conducted is $O(m)$.
\end{proof}

%% file: appendix.tex
\section{Proofs from Section~\ref{sec:verify}}\label{appendix:sec_verify}
\verifyCon*
\begin{proof}[Proof of Theorem~\ref{thm:verify-connectivity-linear}]
Algorithm~\ref{alg:verify-connectivity} decides a fixed permutation
of the given tree's edges and proceeds by doing the following procedure
sequentially on each edge: It queries the given edge until it receives
$c = \lceil \log_{\frac{1-p}{p}} \left( \frac{1}{\delta} \right) \rceil$ 
more positive answers than negative ones, or
until the global predetermined budget of queries is consumed.
If the algorithm runs out of queries then it reports
the tree as \textit{disconnected}, otherwise it reports it 
as \textit{connected}.

Focus on the queries on a given edge $e$ and let $Y_i$ be a random variable
that takes the value $1$ if the $i$-th query is positive and $-1$ otherwise.
Let $X_i$ be the sum of $Y_j$ for the first $i$ queries on edge $e$.
The random variable $X_i$ counts the difference of the amount of positive
and negative answers from queries performed on $e$.
The sequence $\{X_t\}_{t\geq 0}$ is a Markov process and specifically a
biased random walk.
The transition probabilities are the following

\[\makebox[10em][l]{\text{If $e$ is realized:}}
\begin{cases}
\Pr{}{X_{t+1} = X_t + 1} = 1-p \\
\Pr{}{X_{t+1} = X_t - 1} = p \\
\end{cases}
\]

\[\makebox[10em][l]{\text{If $e$ is non-realized:}}
\begin{cases}
\Pr{}{X_{t+1} = X_t + 1} = p \\
\Pr{}{X_{t+1} = X_t - 1} = 1-p \\
\end{cases}
\]

Since $p < \frac{1}{2}$, if the edge $e$ is realized, the random walk is
biased towards positive values and if $e$ does not exist, the random
walk is biased towards negative values. 
Therefore, if we had infinite amount of queries, the random walk of
a realized edge would surpass any positive threshold with probability $1$,
whereas the random walk of a non-existing edge would have 
exponentially small probability of even reaching the threshold. 
However, now that we can perform a finite number
of queries, there is a constant probability that at least one
of the positive edges will fail to reach the threshold.
We prove that the selected threshold and total budget yield an
$(\epsilon, \delta)$ probability guarantee.

First, we calculate the probability to miss-classify a disconnected
tree. Let $A \subseteq E(T),\ A \neq \emptyset,$ be the 
subset of the tree's edges that do not exist. 
In order for Algorithm~\ref{alg:verify-connectivity}
to miss-classify $T$ as connected, the random walk of all the edges in $A$ 
must surpass the threshold.
Hence, the probability of failure is upper bounded by the probability
that a single edge of $A$ surpasses the threshold

\begin{align*}
\Pr{}{\text{ALG\ fails\ on\ } T} 
& = \prod_{e \in A} \Pr{}{e\ \text{passes\ the\ threshold}} \\
& \leq \Pr{}{e \in A\ \text{passes\ the\ threshold}} \\
\end{align*}

We will calculate the probability that a \textit{non-realized} edge $e$ 
reaches the threshold. Let $T_c$ be the threshold hitting time of the edge's 
random walk, that is 
$$
T_c = \inf \left\{ t: X_t = 
\left\lceil \log_{\frac{1-p}{p}} \left( \frac{1}{\delta} \right) \right \rceil
\right\}
$$
The probability that the edge reaches the threshold within the query budget 
is of course bounded by the probability to reach the same threshold 
in infinite amount of steps. 
Therefore the probability of failure is bounded by

$$
\Pr{}{T_c < +\infty}
$$

We define the potential function $\Phi: \mathbb{Z} \to \mathbb{N}$, as
the probability to hit the threshold given that we start 
at a specific location

$$
\Phi(x) = \Pr{}{T_c < +\infty | X_0 = x}
$$

Using a first step analysis and the Markov property we arrive at the following
equations system for $\Phi$

\begin{equation*}
\begin{cases}
\Phi(x) = p \ \Phi(x+1) + (1-p)\ \Phi(x),\ \forall x \in [0,c)\\
\Phi(c) = 1
\end{cases}
\end{equation*}

where $c = \lceil \log_{\frac{1-p}{p}}
\left( \frac{1}{\delta} \right) \rceil$.

From the above system we get that

$$
\Phi(x) = \left( \frac{1-p}{p} \right)^{x-c},\ x\in [0,c]
$$

Given the potential function, the probability of failure is bounded
as follows
$$
\Pr{}{\text{ALG fails on}\ T} \leq \Phi(0) = 
\left( \frac{1-p}{p} \right)^{-\lceil \log_{\frac{1-p}{p}} \left( \frac{1}{\delta} \right) 
\rceil} \leq \delta
$$

Now given a tree $T'$ that is connected, we will calculate the probability
that Algorithm~\ref{alg:verify-connectivity} miss-classifies it as disconnected.
We, again, focus on the random walk of a single edge $e$ and define $T_c$
as previously. Now that the $e$ is realized, we will calculate the expected
number of queries to reach the threshold, by using the below potential 
function.

$$
g(x) = \E{T_c | X_0 = x}
$$

\noindent In a similar manner, we construct the following equations system for $g$

\begin{equation*}
\begin{cases}
g(x) = 1 + (1-p) \ g(x+1) + p\ g(x),\ \forall x \in [0,c)\\
g(c) = 0
\end{cases}
\end{equation*}

\noindent The above system yields the following solution for $g$

$$
g(x) = \tilde{A} \left [  \left( \frac{p}{1-p} \right)^x
- \left( \frac{p}{1-p} \right)^c
\ \right]
+ \frac{x-c}{2p-1}
$$

\noindent for some real constant $\tilde{A}$ and $c = \lceil \log_{\frac{1-p}{p}} \left( \frac{1}{\delta} \right) \rceil$. This gives us the following bound
on the expected value of $T_c$

$$
\E{T_c|X_0=0} = g(0) \leq \frac{c}{1-2p}
$$

Let $\ T_c^1,\ \dots,\ T_c^n\ $ be the hitting times of the random walks
of each edge of $T'$. The only way in which 
Algorithm~\ref{alg:verify-connectivity} can miss-classify $T'$ as disconnected
is if the queries required for all edges to reach the threshold are more than
the available budget. Therefore, using Markov's inequality, the probability
of failure is bounded as follows

\begin{align*}
\Pr{}{\text{ALG\ fails\ on\ } T'} 
& = \Pr{}{S_n \geq \frac{1}{\epsilon} \cdot \frac{c}{1-2p}\cdot n} \\
& \leq \Pr{}{S_n \geq \frac{1}{\epsilon} \cdot \E{S_n}} \\
& \leq \epsilon \\
\end{align*}

\noindent where $S_n = \sum_{i=1}^n T_c^i$
\end{proof}

\section{Proofs from Section~\ref{sec:two_sided}}\label{appendix:sec_2sided}
\naiveTwoSided*

\begin{proof}[Proof of Lemma~\ref{lem:naive-two-sided-high-prob}]
    Fix an edge $e$, and assume that the algorithm performs $k = \log(m^2) / (1 - 2 p)$ queries on it. Then the probability of receiving more true than false query responses for $e$ (ie. of uncovering it correctly in the first step of the algorithm) equals
\begin{align*}
\Pr{}{\text{correctly uncover}\ e} &= 1 - \Pr{}{\text{incorrectly uncover}\ e} \\
&\geq 1 - \Pr{}{X_e < k/2}
\end{align*}
where $X_e \sim \text{Binom}(k, (1-p))$ is the random variable counting the number of true responses for this edge. By using Hoeffding's inequality to bound the tail of a binomial distribution we obtain
\begin{align*}
\Pr{}{\text{correctly uncover}\ e} &\geq 1 - \exp\left(-2 k \left( 1/2 - p \right) \right) \\
&= 1 - \exp\left(-k \left( 1 - 2 p \right) \right) \\
&= 1 - \frac{1}{m^2}
\end{align*}

\noindent Then the probability of correctly uncovering all $m$ edges of the graph is

$$
\Pr{}{\text{correctly uncover all}} \geq \left(1 - \frac{1}{m^2} \right)^m \geq 1 - \frac{1}{m}
$$

where the last step is an application of Bernoulli's inequality. Thus, with high probability, our algorithm succeeds in finding a realized spanning tree.
\end{proof}

\section{Proofs from Section~\ref{sec:FN}}\label{appendix:sec_fn}

\naiveFN*
\begin{proof}[Proof of Lemma~\ref{lem:naive_FN}]
Let $T$ be any realized spanning tree. If an algorithm repeatedly queries all $m$ edges of the moldgraph in a Round-Robin fashion, then by a Coupon Collector's argument on the $n-1$ edges of $T$ it will need $O(\log n)$ rounds of queries on expectation, until it uncovers all edges of the spanning tree. Thus giving us $O(m \log n)$ queries in total.
\end{proof}

\parallelDiscovery*
\begin{proof}[Proof of Lemma~\ref{lemma:parallel-discovery}]
The algorithm proceeds in rounds. In each round it queries one edge from each of the sets of $\mathcal{S}$, performing a total of $k$ queries per round. The probability of getting a False Negative answer for a realized edge of $E_f$ is $p < 1/2$ and the number of queries performed on this edge follows a geometric distribution with parameter $1-p$. Thus, the expected number of queries needed to discover a realized edge of $E_f$ is at most $2 |E_f|$. As the algorithm queries exactly one edge per round, it will run in $2 |E_f|$ rounds, giving a total of $2k|E_f|$ queries in expectation.
\end{proof}

\section{Proofs from Section~\ref{sec:FP}}\label{appendix:sec_fp}

\lemmaDualFP*

\begin{proof}[Proof of Lemma~\ref{lem:dual-alg-fp-tree}]
We will prove Lemma~\ref{lem:dual-alg-fp-tree} by induction
on the number of non-realized edges of $G$. We denote the
edge set of $G$ by $E(G)$, the faces of $G$ by $F(G)$ and the
set of vertices of $G$ by $V(G)$.

\paragraph{Base Case:} Suppose that $G$ only has $1$ non-realized edge,
that is $E(G) = E(T) \cup \{e\}$. We know that the realized subgraph $T$
of $G$ is a tree, therefore $|E(T)|=|V(G)|-1$ and $|E(G)| = |V(G)|$.
As a result, the graph $G$ contains exactly one cycle and two faces;
one that is enclosed by the cycle and the outer face. The dual graph $G'$
will have exactly two vertices, one for each face, which will be connected
by all the edges of the cycle. Therefore, the induced subgraph of $G'$
that contains only the non-realized edge $e$ will be a spanning tree of $G'$.

\paragraph{Induction step:} Assume that for every graph $G$ with at most
$k$ non-realized edges, these edges form a spanning tree of the dual graph $G'$.
We will show that for all graphs $G$ with $k+1$ non-realized edges, these
edges also form a spanning tree of $G'$.

Let $G$ be a graph with $n$ vertices, $n-1$ realized edges that form a
realized tree $T$ and $k+1$ non-realized edges. Let also $e$ be a non-realized
edge of $G$. We remove $e$ from $G$ and take the induced subgraph $G_e$.
The induction step holds for $G_e$, hence its $k$ non-realized edges
form a spanning tree of its dual graph $G_e'$. By adding the edge $e$
in $G_e$, we form a new face $f_e$. We can also see this from Euler's formula:

\begin{align*}
    V(G) + F(G) &= E(G) + 2 \\
                &= E(G_e) + 1 + 2\\
                &= V(G_e) + F(G_e) + 1\\
    & \\
    \Rightarrow F(G) &= F(G_e) + 1\\
\end{align*}

For the new face $f_e$, there are two cases:
\begin{enumerate}
    \item Edge $e$ divides an old face into two new faces, one of them being $f_e$. In this case, the vertex of the old face in the dual graph $G_e'$
    will be expanded into two vertices that are connected by $e$ and,
    using the induction step, the $k+1$ non-realized edges will form 
    a spanning tree of $G'$.
    \item Edge $e$ closes a new cycle and encloses the new face $f_e$.
    The vertex of the face $f_e$ will be connected by $e$ to the vertex
    of the outer face. As a result, the $k+1$ non-realized edges will
    form a spanning tree of $G'$.
\end{enumerate}
\end{proof}

\naiveFP*

\begin{proof}[Proof of Lemma~\ref{lem:naive_FP}]
The algorithm that achieves the guarantees of the theorem 
is similar to the one for the 2-sided error oracle, with the difference that instead of categorizing edges as realized/non-realized based on the majority of query responses, we consider an edge as non-realized when it receives at least one ``No" response, and as realized otherwise. In this case we perform $k = \log_2(m^2)$ queries on each edge. Thus, to bound the probability of failure for a single edge it suffices to bound the probability of a non-realized edge giving $k$ false ``Yes" responses. This equals

$$
\Pr{}{\text{uncorrectly uncover}\ e} \leq p^k < \left(\frac{1}{2}\right)^k = \frac{1}{m^2}
$$

The rest of the proof is the same as for the 2-sided error oracle, and gives us that the algorithm succeeds with probability $(1 - 1/m)$.
\end{proof}

\begin{figure}[H]
\centering
\begin{subfigure}{\linewidth}
    \centering
    \includegraphics[width=0.73\linewidth]{pics/moldgraph.png}
    \caption{The moldgraph used for the proof of Theorem~\ref{thm:1-sided-fp-lb}}
    \label{subfig:FP-moldgraph}
\end{subfigure}%
\newline
\begin{subfigure}{\linewidth}
    \centering
    \includegraphics[width=0.73\linewidth]{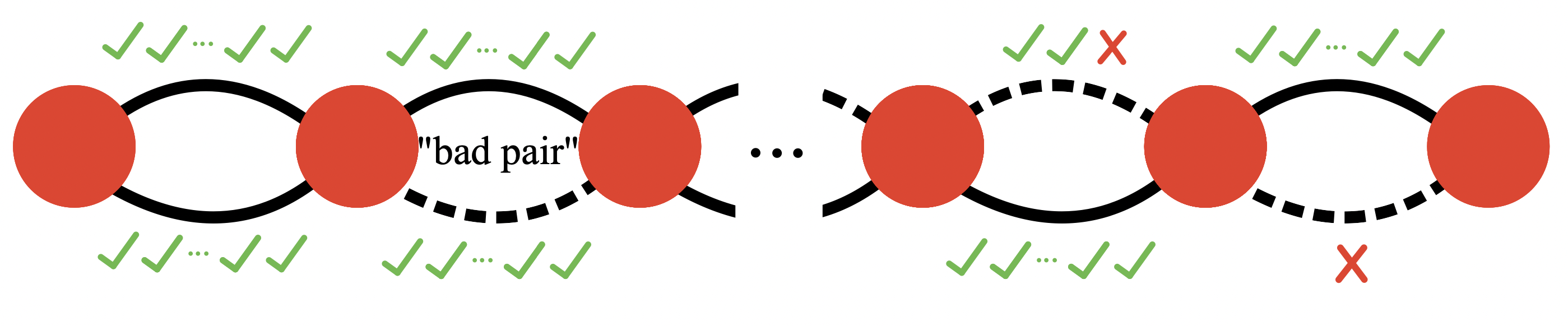}
    \caption{A realized graph from the family mentioned in the proof of Theorem~\ref{thm:1-sided-fp-lb}. The dotted lines are non-realized edges
    and by the edges lie the answers of the 1-sided FP noisy oracle.}
    \label{subfig:FP-LB}
\end{subfigure}%

\caption{}
\label{fig:FP-LB}
\end{figure}

\oneSidedFPLB*
\begin{proof}[Proof of Theorem~\ref{thm:1-sided-fp-lb}]
The outline of this proof is very similar to that of Lemma~\ref{lem:2-sided-error-lb}. Consider the graph in Figure~\ref{fig:FP-LB} which has $n+1$ vertices and $2n$ edges. Between any two consecutive vertices exist 2 parallel edges. Half of the edge-pairs are arbitrarily picked to have only 1 realized edge, while in the other half both of their edges are realized.

Every time that we receive a negative query response for an edge, we can be certain that this edge is not realized and thus discard it. Now, call ``bad pair" an edge-pair having one non-realized edge for which the algorithm receives $k$ consecutive positive query responses. We set $k = \Theta(\log(n / \log n)) = \Theta(\log n)$, which gives us $\Pr{}{\text{bad}} \geq \log n / n$. Thus the probability that \textit{at least} one such ``bad pair" exists is

$$
\Pr{}{\text{bad exists}} \geq 1 - \left( 1 - \frac{\log n}{n} \right)^{n/2} \geq 1 - \frac{1}{n}
$$

Note that this ``bad pair" should have received the exact same responses as edge-pairs having two realized edges. There are $n/2 = \Theta(n)$ such pairs and the ``bad pair" will be indistinguishable from them, as they have been picked arbitrarily. Thus, by a similar argument as in Theorem~\ref{lem:2-sided-error-lb} we conclude that any algorithm with constant probability of failure requires $\Omega(n \log n)$ queries to discover a spanning tree in this specific graph instance.

The desired $\Omega(m \log m)$ lower bound again results from the fact that if we had an algorithm performing better than $\Omega(m \log m)$ queries, it would perform better than $\Omega(n \log n)$ in this instance, as it has $m = O(n)$.
\end{proof}

%% file: main.bbl
\newcommand{\etalchar}[1]{$^{#1}$}
\begin{thebibliography}{HEK{\etalchar{+}}08}

\bibitem[BCE80]{bollobas1980hadwiger}
B{\'e}la Bollob{\'a}s, Paul~A Catlin, and Paul Erd{\"o}s.
\newblock Hadwiger's conjecture is true for almost every graph.
\newblock {\em Eur. J. Comb.}, 1(3):195--199, 1980.

\bibitem[EH14]{ErleHoff014}
Thomas Erlebach and Michael Hoffmann.
\newblock Minimum spanning tree verification under uncertainty.
\newblock In Dieter Kratsch and Ioan Todinca, editors, {\em Graph-Theoretic
  Concepts in Computer Science - 40th International Workshop, {WG} 2014,
  Nouan-le-Fuzelier, France, June 25-27, 2014. Revised Selected Papers}, volume
  8747 of {\em Lecture Notes in Computer Science}, pages 164--175. Springer,
  2014.

\bibitem[FFX{\etalchar{+}}17]{FuFuXuPengWangLu2017}
Luoyi Fu, Xinzhe Fu, Zhiying Xu, Qianyang Peng, Xinbing Wang, and Songwu Lu.
\newblock Determining source-destination connectivity in uncertain networks:
  Modeling and solutions.
\newblock {\em {IEEE/ACM} Trans. Netw.}, 25(6):3237--3252, 2017.

\bibitem[FRPU94]{upfal_noisy_94}
Uriel Feige, Prabhakar Raghavan, David Peleg, and Eli Upfal.
\newblock Computing with noisy information.
\newblock {\em SIAM Journal on Computing}, 23(5):1001--1018, 1994.

\bibitem[FWK14]{FuWangKuma2014}
Luoyi Fu, Xinbing Wang, and P.~R. Kumar.
\newblock Optimal determination of source-destination connectivity in random
  graphs.
\newblock In Jie Wu, Xiuzhen Cheng, Xiang{-}Yang Li, and Saswati Sarkar,
  editors, {\em The Fifteenth {ACM} International Symposium on Mobile Ad Hoc
  Networking and Computing, MobiHoc'14, Philadelphia, PA, USA, August 11-14,
  2014}, pages 205--214. {ACM}, 2014.

\bibitem[HEK{\etalchar{+}}08]{HoffErleKrizMihalRama2008}
Michael Hoffmann, Thomas Erlebach, Danny Krizanc, Mat{\'{u}}s Mihal{\'{a}}k,
  and Rajeev Raman.
\newblock Computing minimum spanning trees with uncertainty.
\newblock In Susanne Albers and Pascal Weil, editors, {\em {STACS} 2008, 25th
  Annual Symposium on Theoretical Aspects of Computer Science, Bordeaux,
  France, February 21-23, 2008, Proceedings}, volume~1 of {\em LIPIcs}, pages
  277--288. Schloss Dagstuhl - Leibniz-Zentrum f{\"{u}}r Informatik, Germany,
  2008.

\bibitem[Kos84]{kostochka1984lower}
Alexandr~V. Kostochka.
\newblock Lower bound of the hadwiger number of graphs by their average degree.
\newblock {\em Combinatorica}, 4(4):307--316, 1984.

\bibitem[Lov06]{lovasz2006}
L{\'a}szl{\'o} Lov{\'a}sz.
\newblock Graph minor theory.
\newblock {\em Bulletin of the American Mathematical Society}, 43(1):75--86,
  2006.

\bibitem[LRAZ21]{LyuRenAbbaZhan2021}
Jianhua Lyu, Yiran Ren, Zeeshan Abbas, and Baili Zhang.
\newblock Reliable route selection for wireless sensor networks with connection
  failure uncertainties.
\newblock {\em Sensors}, 21(21):7254, 2021.

\bibitem[LW70]{lick_white_1970}
Don~R. Lick and Arthur~T. White.
\newblock k-degenerate graphs.
\newblock {\em Canadian Journal of Mathematics}, 22(5):1082–1096, 1970.

\bibitem[MS17]{mazumdar_clustering_2017}
Arya Mazumdar and Barna Saha.
\newblock Clustering with noisy queries.
\newblock In I.~Guyon, U.~Von Luxburg, S.~Bengio, H.~Wallach, R.~Fergus,
  S.~Vishwanathan, and R.~Garnett, editors, {\em Advances in Neural Information
  Processing Systems}, volume~30. Curran Associates, Inc., 2017.

\bibitem[RS84]{robertson1984graph}
Neil Robertson and Paul~D Seymour.
\newblock Graph minors. iii. planar tree-width.
\newblock {\em Journal of Combinatorial Theory, Series B}, 36(1):49--64, 1984.

\bibitem[STC09]{SchmToppCreme2009}
Frank~R. Schmidt, Eno T{\"{o}}ppe, and Daniel Cremers.
\newblock Efficient planar graph cuts with applications in computer vision.
\newblock In {\em 2009 {IEEE} Computer Society Conference on Computer Vision
  and Pattern Recognition {(CVPR} 2009), 20-25 June 2009, Miami, Florida,
  {USA}}, pages 351--356. {IEEE} Computer Society, 2009.

\bibitem[YF15]{YarkFowl2015}
Julian Yarkony and Charless~C. Fowlkes.
\newblock Planar ultrametrics for image segmentation.
\newblock In Corinna Cortes, Neil~D. Lawrence, Daniel~D. Lee, Masashi Sugiyama,
  and Roman Garnett, editors, {\em Advances in Neural Information Processing
  Systems 28: Annual Conference on Neural Information Processing Systems 2015,
  December 7-12, 2015, Montreal, Quebec, Canada}, pages 64--72, 2015.

\end{thebibliography}
